\documentclass[a4paper,11pt,]{article}
\usepackage[dvipdfmx]{graphicx}
\usepackage{tabularx}
 \usepackage{amssymb}
 \usepackage{amsmath}
\usepackage{multirow}

\newtheorem{THM}{Theorem}[section]

\global \count10 = 0

\global \count11 = 1

\global \count12 = 1

\global \count13 = 1

\newcommand{\com}[1]{\ifnum\count13<1 #1 \fi}

\def\squarebox#1{\hbox to #1{\hfill\vbox to #1{\vfill}}}
\def\qed{\hspace*{\fill}%
        \vbox{\hrule\hbox{\vrule\squarebox{.667em}\vrule}\hrule}\smallskip}
\newenvironment{proof}{\begin{trivlist}
\item[\hspace{\labelsep}{\em\noindent Proof.~}]}{\qed\end{trivlist}}

\def\squarebox#1{\hbox to #1{\hfill\vbox to #1{\vfill}}}
\def\qed{\hspace*{\fill}%
        \vbox{\hrule\hbox{\vrule\squarebox{.667em}\vrule}\hrule}\smallskip}

\begin{document}

\title{
	An improved lower bound \\ for one-dimensional online unit clustering
} 

\author{
	Jun Kawahara$^{1}$ and 
	Koji M. Kobayashi$^{2}$ 
	\\   
	{\footnotesize 
		$^{1}$Graduate School of Information Science, Nara Institute of Science and Technology 
	}
	\\
	{\footnotesize 
		$^{2}$National Institute of Informatics 
	}
}

\date{}
\maketitle

\begin{abstract}
	The online unit clustering problem was proposed by Chan and Zarrabi-Zadeh (WAOA2007 and Theory of Computing Systems 45(3), 2009), 
	which is defined as follows:
	``Points'' are given online in the $d$-dimensional Euclidean space one by one.
	An algorithm creates a ``cluster,'' which is a $d$-dimensional rectangle. 
	The initial length of each edge of a cluster is 0. 
	An algorithm can extend an edge until it reaches unit length independently of other dimensions.
	The task of an algorithm is to cover a new given point either by creating a new cluster and assigning it to the point, or by extending edges of an existing cluster created in past times. 
	The goal is to minimize the total number of created clusters.
	Chan and Zarrabi-Zadeh proposed some method to obtain a competitive algorithm for the $d$-dimensional case 
	using an algorithm for the one-dimensional case, 
	and thus the one-dimensional case has been extensively studied including some variants of the unit clustering problem. 

	In this paper, 
	we show a lower bound of $13/8 = 1.625$ on the competitive ratio of any deterministic online algorithm for the one-dimensional unit clustering, 
	improving the previous lower bound $8/5 (=1.6)$ presented by Epstein and van Stee (WAOA2007 and ACM Transactions on Algorithms 7(1), 2010). 
	Note that Ehmsen and Larsen (SWAT2010 and Theoretical Computer Science, 500, 2013) showed the current best upper bound of $5/3$, and 
	conjectured that the exact competitive ratio in the one-dimensional case may be $13/8$. 
\end{abstract}

\section{Introduction} \label{Intro}

Given some points, 
we consider the problem of partitioning points into some groups, namely {\em clusters}, 
which is generally called a {\em clustering problem}. 
The goal is to optimize an objective function according to applications. 
Clustering problems are applied for various purposes 
such as information retrieval, data mining and facility location. 
%
%
This problem is formulated as online problems, 
and some variants have been widely studied.
The {\em online unit covering problem}, proposed by Charikar et al.~\cite{MC04}, is the most basic variant.
In this problem, points are given in the $d$-dimensional Euclidean space one by one.
An algorithm places balls with unit radius so that they cover points.
The goal of the problem is to minimize the number of created balls. 
Chan and Zarrabi-Zadeh have introduced the {\em online unit clustering problem} \cite{TC07}.
In this problem, an {\em input} is a sequence of points.
Points are given in the $d$-dimensional Euclidean space one by one
and must be covered by rectangles, called {\em clusters}.
Initially, the length of each edge of a cluster is 0, and
each edge can be extended until unit length.
Once a cluster is placed, it must not be shifted and removed.
When a point is given outside existing clusters,
the task of an algorithm is to determine whether creating a new cluster or extending
some edges of an existing cluster to cover the point before the next one is given.
The goal of the problem is to minimize the number of created clusters. 
The performance of an online algorithm is evaluated using competitive analysis \cite{AB98,DS85}.
For any algorithm $ALG$ and an input $\sigma$, 
let $C_{ALG}(\sigma)$ denote the number of clusters created by $ALG$. 
If, for any input $\sigma$, an online algorithm $ON$ creates clusters at most $c$ of the optimal offline algorithm, called $OPT$, for $\sigma$, 
then we say that $ON$ is $c$-competitive or that the competitive ratio of $ON$ is at most $c$. 
Chan and Zarrabi-Zadeh \cite{TC07} showed that a $2^{d-1} c$-competitive algorithm can be obtained by using a $c$-competitive algorithm for the one-dimensional case (i.e., the case of $d = 1$), which is the current best algorithm for the $d$-dimensional case when $d \geq 2$.
This implies that by improving an algorithm for the one-dimensional case, we can also obtain a
better algorithm for the $d$-dimensional case.
Thus, the one-dimensional case has been extensively studied including some variants of the unit clustering problem. 
(The past results are summarized in Table~\ref{tab:prevresult}.)

\noindent
{\bf Previous Results and Our Result.}~~~
For the one-dimensional case, 
the best known deterministic upper bound on the competitive ratio is $5/3 \ (\approx 1.667)$, presented by
Ehmsen and Larsen \cite{ME10}, and the best known lower bound on the competitive ratio of any deterministic online algorithm
is $8/5 \ (=1.6)$, shown by Epstein and van Stee \cite{LE07}.

In this paper, for the one-dimensional deterministic case, 
we improve the lower bound from $8/5$ to $13/8 = 1.625$.
(Ehmsen and Larsen \cite{ME10} conjectured that the exact competitive ratio may be $13/8$.)

\begin{table}[t]
  \begin{center}
    \caption{Upper and lower bounds for one-dimensional unit clustering}
    \begin{tabular}{llll}
    \hline
                & Years             & Deterministic           & Randomized \\
                & (first appearance)&                         & \\ \hline
    Upper Bound & 2007 \cite{TC07}  & $2$                     & $15/8 \ (= 1.875)$ \\
                & 2007 \cite{HZ07}  &                         & $11/6 \ (\approx 1.834)$ \\
                & 2008 \cite{LE07}  & $7/4 \ (= 1.75)$        & $7/4 \ (= 1.75)$ \\
                & 2010 \cite{ME10}  & $5/3 \ (\approx 1.667)$ & \\ \hline
    Lower Bound & 2015 [this paper] & $13/8 \ (= 1.625)$      & \\
                & 2008 \cite{LE07}  & $8/5 \ (=1.6)$          & $3/2 \ (= 1.5)$ \\
                & 2007 \cite{TC07}  & $3/2 \ (=1.5)$          & $4/3 \ (\approx 1.333)$ \\
    \hline
    \end{tabular}
    \label{tab:prevresult}
  \end{center}
\end{table}

\noindent
{\bf Related Results.}~~~
For the two-dimensional case, 
Epstein and van Stee \cite{LE07} gave a deterministic lower bound of $2$ and a randomized lower bound of $11/6 \ (\approx 1.833)$. 
Ehmsen and Larsen \cite{ME10} showed a lower bound of $13/6 \ (\approx 2.166)$ for any deterministic algorithm. 
Epstein et~al.\ \cite{LE08} and Csirik et~al.\ \cite{JC13} studied some variants of the one-dimensional online unit clustering. 

\section{Improved Lower Bound} \label{LB}

\begin{THM} \label{thm:1}
	The competitive ratio of any deterministic online algorithm is at least $13/8$. 
\end{THM}

\begin{proof}
	Let $ON$ be any deterministic online algorithm and we give $ON$ an instance $\sigma$ in Table~\ref{tab:lb}. 
	In the table, 
	column ``Point'' presents the position of each given new point. 
	``Cluster'' shows the name of $ON$'s cluster assigned to the new point. 
	If there exist other clusters which can be assigned at that time, 
	column ``Note'' explains how $ON$ behaves and how much the costs of $ON$ and $OPT$ are if the other clusters are used. 
	In addition, 
	``Note'' tells that $ON$ can assign an existing cluster $K$ or $L$ to the last given point and 
	how much the costs are charged at each case. 
	%
	For example, when $ON$ is given the fourth point at position 6, 
	$ON$ can choose either to create a new cluster $F$ or to assign the existing cluster $E$. 
	If $ON$ creates $F$, then the next point arrives at position 2 and the input goes on. 
	Otherwise, three points are given one by one and the cost ratio is $5/3$, which is larger than $13/8$. 

	Therefore, 
	we can show that $C_{ON}(\sigma) / C_{OPT}(\sigma) \geq 13/8$. 
\begin{table*}[h]
\begin{center}
	\renewcommand{\arraystretch}{1.1}
	\caption{Lower Bound Instance $\sigma$}
	\begin{tabular}{llp{120mm}}
		\hline
			Point & Cluster &  Note \\
		\hline
			3 & $D$ &  \\
		\hline
			4 & $D$ & If $ON$ creates a new cluster, then we have $C_{ON}(\sigma) = 2$ and $C_{OPT}(\sigma) = 1$. \\
		\hline
			5 & $E$ &  \\
		\hline
			6 & $F$ & If $ON$ assigns $E$ to the given point, then three points are given at $2.5, 4.5$ and $6.5$ respectively and $C_{ON}(\sigma) = 5$ and $C_{OPT}(\sigma) = 3$ hold. \\
		\hline
			2 & $B$ &  \\
		\hline
			1 & $B$ & If a new cluster is created, we obtain the cost ratio of $5/3 \approx 1.66$. \\
		\hline
			0 & $A$ &  \\
		\hline
			2.5 & $C$ &  \\
		\hline
			7 & $F$ & If $ON$ creates a new one, then the cost ratio becomes $7/4 = 1.75$. \\
		\hline
			8 & $G$ & \\
		\hline
			8.5 & $H$ & If $G$ is assigned to the given point, then four points arrive at $4.5, 5.6, 7.4$ and $9.5$ respectively and we have the ratio of $10/6 \approx 1.66$. ((i) in Fig.~\ref{fig:LB})\\
		\hline
			9 & $H$ & If $ON$ newly creates a cluster, we obtain the cost ratio of $9/5 = 1.8$. \\
		\hline
			10 & $J$ &  \\
		\hline
			11 & $J$ & If a new cluster is created, then $C_{ON}(\sigma) = 10$ and $C_{OPT}(\sigma) = 6$ hold. \\
		\hline
			9.5 & $I$ & \small If $ON$ assigns $H$ to the given point, then six points are given at $4.5, 5.6, 7.2, 8.3, 9.7$ and $11.5$ respectively, and we have the cost ratio of $13/8 = 1.625$. ((ii) in Fig.~\ref{fig:LB})\\
		\hline
			12 & $K$ &  \\
		\hline
			13 & $K$ & Two points arrive at $11.5$ and $14$ respectively and the cost ratio becomes $13/8 = 1.625$. ((iii) in Fig.~\ref{fig:LB})\\
		\cline{2-3}
			   & $L$ & Two points are given at $4.5$ and $5.6$ respectively and the ratio of $13/8 = 1.625$ is obtained. ((iv) in Fig.~\ref{fig:LB})\\
		\hline
	\end{tabular}
	\label{tab:lb}
\end{center}
\end{table*}

\end{proof}
\begin{figure*}[h]
	 \begin{center}
	  \includegraphics[width=140mm]{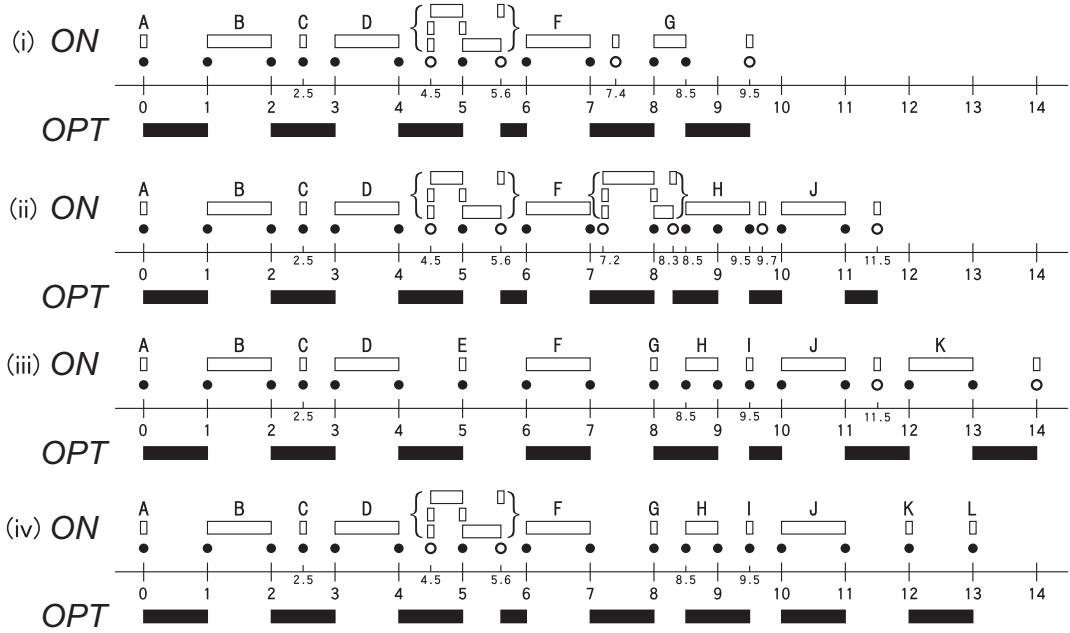}
	 \end{center}
	 \caption{
	 	States of Clusters. 
	 	White rectangles above the axes denote clusters created by $ON$. 
		On the other hand, black ones mean $OPT$'s. 
		For example, 
		the top figure (i) demonstrates that 
		$ON$ assigns $G$ to the point at 8.5 and then the four points 4.5, 5.6, 7.4 and 9.5 are given. 
		Note that $ON$ uses at least two clusters to cover the three points 4.5, 5 and 5.6. 
	 }
	\label{fig:LB}
\end{figure*}

\section{Discussion}
\ifnum \count10 > 0

\fi
\ifnum \count11 > 0
Developing an online algorithm $ON$ whose competitive ratio is at most $c$ in the one-dimensional case, 
we have little other choice than showing that 
for each $x$ clusters $ON$ creates, $OPT$ necessarily does at least $x/c \ (= y)$ clusters. 
This also appears in the analyses of the $7/4$ and $5/3$-competitive algorithms in \cite{LE07} and \cite{ME10}, respectively.
Ehmsen and Larsen \cite{ME10} stated that 
if $x \leq 20$ and the optimal competitive ratio for the one-dimensional case is neither $8/5$ nor $5/3$, 
then the possible competitive ratio is only $13/8$ or $18/11$. 
They conjectured that these values may be optimal. 
(Note that pairs of integers $(x, y)$ such that $8/5 < x/y < 5/3$ holds are
$(13, 8),\ (18, 11),\ (21, 13),\allowbreak \ (29, 18),\allowbreak \ (34, 21),\ldots$ in ascending order of $x$.) 
Generally speaking, we need to carry out a more exhaustive case-analysis 
as $x$ grows or $c$ decreases. 
Thus, to obtain $ON$'s competitive ratio of $c=13/8$, i.e., $x = 13$,
it seems that more complicated analysis and $ON$'s behaviors are required than $c=5/3$. 
Specifically, given our lower bound instance,
$ON$ must behave in the same way described in Table~\ref{tab:lb}.
(Of course $ON$ has to be able to deal with the other inputs as well.) 

\fi

\end{document}